\documentclass[12pt]{article}

\usepackage{amsmath,amssymb,amsthm, units, color}

\usepackage{graphicx}

\usepackage[all]{xy}

\usepackage{tikz}

\newtheorem{definition}{Definition}[section]

\newtheorem{lemm}{Lemma}[section]

\newtheorem{exam}{Example}[section]

\newtheorem{prop}{Proposition}[section]

\newtheorem{cor}{Corollary}[section]

\newtheorem{theo}{Theorem}[section]

\def\<{\left <}

\def\>{\right >}

\def\({\left (}

\def\){\right )}

\def\cbra{\left \{}

\def\cket{\right \}}

\DeclareSymbolFont{AMSb}{U}{msb}{m}{n}

\DeclareMathSymbol{\N}{\mathbin}{AMSb}{"4E}

\DeclareMathSymbol{\Z}{\mathbin}{AMSb}{"5A}

\DeclareMathSymbol{\R}{\mathbin}{AMSb}{"52}

\DeclareMathSymbol{\Q}{\mathbin}{AMSb}{"51}

\DeclareMathSymbol{\I}{\mathbin}{AMSb}{"49}

\DeclareMathSymbol{\C}{\mathbin}{AMSb}{"43}

\newcommand{\weg}[1]{}

\title{A colouring protocol for the generalized Russian cards problem}

\author{Andr\'es Cord\'on--Franco, Hans van Ditmarsch, \\ David Fern\'andez--Duque, and Fernando Soler--Toscano\thanks{Emails and affiliations: {\tt \{acordon,hvd,dfduque,fsoler\}@us.es}, University of Sevilla, Spain. Hans van Ditmarsch is also affiliated to IMSc, Chennai, India, as a research associate.}}

\begin{document}

\maketitle

\begin{abstract}
In the generalized Russian cards problem, Alice, Bob and Cath draw $a$, $b$ and $c$ cards, respectively, from a deck of size $a+b+c$. Alice and Bob must then communicate their entire hand to each other, without Cath learning the owner of a single card she does not hold. Unlike many traditional problems in cryptography, however, they are not allowed to codify or hide the messages they exchange from Cath. The problem is then to find methods through which they can achieve this. We propose a general four-step solution based on finite vector spaces, and call it the ``colouring protocol'', as it involves colourings of lines.

Our main results show that the colouring protocol provides a solution to the generalized Russian cards problem in cases where $a$ is a power of a prime, $c=O(a^2)$ and $b=O(c^2)$. This improves substantially on the collection of parameters for which solutions are known. In particular, it is the first solution which allows the eavesdropper to have more cards than one of the communicating players.
\end{abstract}


\section{Introduction}

The goal of this article is to provide a solution to the {\em generalized Russian cards problem} \cite{albertetal:2005} based on finite vector spaces. The problem is parametrized by the triple $(a,b,c)$, which we will often call its {\em size} and may be stated as follows:

\paragraph{The generalized Russian cards problem}
\begin{quote}
{Alice, Bob and Cath each draw $a,b$ and $c$ cards, respectively, from a deck of $a+b+c$. All players know which cards were in the deck and how many of them the other players drew, but each player may only see the cards in their own hand.

Alice and Bob, however, want to know exactly which cards the other holds. Moreover, they do not want for Cath to learn who holds {\em any} card whatsoever, aside of course from her own cards.

However, they may only do so by making unencrypted public announcements, so that Cath can learn all the information that they exchange.

Can Alice and Bob achieve this?}
\end{quote}
\vskip 15pt 

Some general assumptions are needed to make the problem precise, which we shall formalize in Section \ref{sec.some}. First, the cards are dealt beforehand in a secure phase which we treat as a black box and gives the players no information about others' cards. The agents have no communication before this phase and cannot share secrets (such as private keys). Second, the agents have unlimited computational capacity.\footnote{This assumption is unusual in traditional cryptography but is standard in information-theoretic approaches where {\em unconditional} rather than {\em probabilistic} security is sought \cite{maurer:1999}.} This means, on one hand, that solutions via encryption are not valid, provided they are vulnerable to cryptanalisis (independently of its computational cost). It also means that we shall not be concerned with the feasibility of agents' strategies, although this is certainly an important line for future inquiry. Third, all strategies are public knowledge, in keeping with Kerckhoffs' principle \cite{kerckhoffs:1883}. These assumptions turn the problem into a challenging combinatorial puzzle.

A (possible) solution to the generalized Russian cards problem is a {\em protocol;} these shall also be discussed formally in Section \ref{sec.some}. The solution we propose is, to the best of our knowledge, the first solution to the generalized Russian cards problem which works in cases where Cath holds more cards than Alice.

The Russian cards problem itself originates with Kirkman \cite{kirkman:1847}. There, the solution takes the form of a {\em design}, a collection of subsets of a given set that satisfies certain regulaties \cite{stinson:2004}. The design consists of seven triples. Cryptography based on card deals is also investigated in various other (not necessarily related) publications, such as \cite{mizukietal:2002,stiglic:2001}.

Many instances of the generalized Russian cards problem have been studied. We remark that all of these protocols obtain a notion of security called {\em weak $1$-security} by Swanson and Stinson \cite{swanson:2012}, the notion we are concerned with in this paper. A stronger and more general notion of {\em perfect $k$-security} is also considered by them but we shall not treat it here. The known solutions are as follows.
\begin{enumerate}
\item The original $(3,3,1)$ case has been extensively studied and has many solutions \cite{kirkman:1847,makarychevs:2001}.
\item The case $(4,4,2)$ has a three-step solution given in \cite{threesteps}, where it is also shown that there is no two-step solution.
\item All cases $(a,2,1)$ provided $a\equiv 0,4\pmod 6$ have been solved in \cite{albertetal:2005}.
\item  All cases $(a,b,1)$ with both $a,b>2$ have been solved in \cite{cordonetal:2012}.
\item If $c+1<a$ then there is a solution for $b=O(a^2)$ \cite{albertetal:2005}.
\end{enumerate}
The case $(a,b,1)$ with $a,b > 2$ subsumes the case $(3,b,1)$ for $b \geq 3$, also covered previously in \cite{albertetal:2005}; \cite{swanson:2012} demonstrates that, for the stronger notion of perfect 1-security, every announcement in a $(3,b,1)$ solution is a Steiner triple. They have yet another characterization result for perfect $k$-security. 

Our main contribution is to provide solutions for infinitely many cases where $b=O(c^2)$ and either $c=O(a^{\nicefrac 32})$ or $c=O(a^2)$, the first known solutions for $c>a$. We remark that all but one of the previously known protocols are two-step protocols, whereas our protocol is in four steps. In \cite{albertetal:2005} it is shown that there can be no two-step solution if $c\geq a-1$. As we provide solutions with $c>a$, the protocol we propose is longer. We do not know if a shorter, three-step variant is possible.

The plan of the paper is as follows. Section \ref{sec.some} formalizes our model of security and our notion of protocol. Section \ref{seccolp} gives an informal description of the protocol, in order to motivate the more specific description given in Definition \ref{def.colprotocol}. This protocol depends on certain parameters which do not exist for all values of $(a,b,c)$. The conditions under which such parameters exist are treated in Section \ref{line}. Finally, Section \ref{suitable} computes explicit bounds on the parameters for which the protocol is executable.

\section{Protocols and security} \label{sec.some} 

In this section we shall present the notion of {\em protocol} we shall use. Throughout this paper, we assume that $D$ is a fixed, finite set of ``cards''. A {\em card deal} is a partition $(A,B,C)$ of $D$; the deal has {\em size} $(a,b,c)$ if $A$ is an $a$-set, $B$ a $b$-set and $C$ a $c$-set, where by ``$x$-set'' we mean a set of cardinality $x$. The set of $x$-subsets of $Y$ is written $Y\choose x$. We think of $A$ as the {\em hand} of Alice, or that Alice {\em holds $A$}; similarly, $B$ and $C$ are the hands of Bob and Cath, respectively. In general we may simply assume that $D=\{1,\hdots,a+b+c\}$, and define ${\rm Deal}(a,b,c)$ to be the set of partitions of $D$ of size $(a,b,c)$.

Central to this text is the notion of {\em protocol.} Roughly, a protocol is a family of rules or instructions that Alice and Bob must follow in order to send out a sequence of public announcements. Protocols are in principle non-deterministic, assigning only a probability measure to the announcements that an agent may make.

In \cite{swanson:2012} and other papers, an announcement has been modelled as a set of hands that one of the agents may hold. Thus Alice would announce a subset $\mathcal A$ of $D\choose a$, indicating that $A\in\mathcal A$. More specifically, an announcement is a pair $(i,\mathcal A)$, where $i$ is the agent making the announcement (either Alice or Bob). But for our proposed solution, agents' announcements are not in one-one correspondence to subsets of $D\choose a$. Because of this, it will be simpler to model announcements as elements of a countable or finite set, say $\mathbb N$; a similar setup is proposed in \cite{swanson:2012}, where Alice enumerates all subsets of $D\choose a$ by $\mathcal A_0,\mathcal A_1,\hdots,\mathcal A_n$ and merely announces the appropriate subindex.

Thus we let ${\rm Ann}=\{Alice,Bob\}\times J$ be the set of announcements, where $J\subseteq\mathbb N$. A sequence of announcements $\vec\alpha=(i_0,\phi_0),(i_1,\phi_1),\hdots,(i_n,\phi_n)$ is a {\em run}. We shall write ${\rm Run}$ for the set of runs, including the empty sequence.

An (unusual) assumption of the problem is that there is a secure dealing phase which we treat as a black box. Initially, a card deal is selected randomly and players have knowledge of their own cards and of the size $(a,b,c)$ of the deal, but know nothing more about others' cards. Thus they are not able to distinguish between different deals where they hold the same hand. We model this by equivalence relations between deals; since from Alice's perspective, $(A,B,C)$ is indistinguishable from $(A,B',C')$, we define $(A,B,C)\stackrel{Alice}\sim(A',B',C')$ if and only if $A=A'$. We also define analogous equivalence relations for Bob and Cath.

In \cite{swanson:2012}, Swanson and Stinson define strategies of length two as assigning a probability distribution to Alice's possible announcements (after which Bob announces Cath's hand), and consider $n$-equitable strategies as a special case where Alice always has $n$ possible announcements each with probability $\nicefrac 1n$. We will work with {\em equitable} strategies (and merely call them {\em protocols}), where probability is distributed evenly among the possible outcomes but the number of outcomes is left unspecified. This will allow us to dispense with probability measures and only determine a set of possible announcements at each stage of the protocol. On the other hand, we consider protocols which may have more than two steps, making our definition more elaborate than Swanson and Stinson's.

\begin{definition}[Protocol]\label{defprot}
Let ${\rm Deal}={\rm Deal}(a,b,c)$.

A {\em protocol} (with $(a,b,c)$ as parameters) is a pair of functions $(j,\pi)$ assigning to every deal $\delta\in{\rm Deal}$ and every run $\vec\alpha\in{\rm Run}$ an agent $j(\vec\alpha)\in \{Alice,Bob\}$ and a finite set $\pi(\delta,\vec\alpha)\subseteq \mathbb N$ such that if $\delta'\stackrel {j(\vec\alpha)}\sim \delta$, then $\pi(\delta',\vec\alpha)=\pi(\delta,\vec\alpha)$.
\end{definition}

Thus once a deal has been given, a protocol assigns to each run a player who is to make the next announcement and a set of possible announcements for the player to make. These announcements are determined exclusively by the information an agent has access to, which is assumed to be {\em only} ({\it i}) their hand, ({\it ii}) the parameters $a,b,c$, ({\it iii}) the announcements that have been made previously and ({\it iv}) the protocol being executed.

Protocols are non-deterministic in principle and hence may be executed in many ways; an {\em execution of a protocol} is a pair $(\delta,\vec\alpha)$, where $\delta$ is a deal, $\vec\alpha=(i_0,\phi_0),\hdots,(i_n,\phi_n)$ a run and, for all $k<n$, $i_{k+1}=j\big((i_0,\phi_0),\hdots,(i_k,\phi_k)\big)$ and $\phi_{k+1}\in \pi\big(\delta,(i_0,\phi_0),\hdots,(i_k,\phi_k)\big).$

The first property that a protocol needs to have in order to be successful is for its set of executions to be non-empty. This observation seems trivial but it will occupy a large portion of the paper.

\begin{definition}[Executability]
Given a natural number $N$, a protocol is {\em $N$-executable} if for every deal $\delta$ there is an execution $(\delta,\vec\alpha)$ where $\alpha$ has length $N$.
\end{definition} 

We may omit the parameter $N$ and let it be given by context; the protocol we propose has four steps so we take $N=4$.

The second property that a protocol must have to be successful is that Alice and Bob know each other's cards (and hence the entire deal) after its execution:

\begin{definition}[Informativity]
An execution $((A,B,C),\vec\alpha)$ is informative for Alice if there is no execution $((A,B',C'),\vec\alpha)$ with $C'\not=C$.

Similarly, an execution $((A,B,C),\vec\alpha)$ is informative for Bob if there is no execution $((A',B,C'),\vec\alpha)$ with $C'\not=C$.

A protocol is informative if there exists $N$ such that the protocol is $N$-executable and every execution of length $N$ is informative both for Alice and for Bob.
\end{definition}

The third property is that, given a card $x$ which Cath does not hold, she should consider it possible that either Alice holds it or Bob holds it:

\begin{definition}[Safety]\label{w1s}
An execution $((A,B,C),\vec\alpha)$ of a protocol $(j,\pi)$ is {\em safe} if for every $x \not\in C$ there is
\begin{enumerate}
\item  a deal $\delta'=(A',B',C)$ such that $x\in A'$ and $(\delta',\vec\alpha)$ is also an execution of $(j,\pi)$, as well as
\item a deal $\delta''=(A'',B'',C)$ such that $x\in B''$ and $(\delta'',\vec\alpha)$ is also an execution of $(j,\pi)$.
\end{enumerate}
The protocol $(j,\pi)$ is safe if every execution of $(j,\pi)$ is safe.
\end{definition}

Safety is equivalent to {\em weak $1$-security} as defined by Swanson and Stinson \cite{swanson:2012}, but this is not the only notion of security they discuss. Weak security is equivalent to the statement that for all $C$, $x\not\in C$ and every execution $\vec\alpha$ of the protocol,
\[0<{\sf Pr}(x\in A|C,\vec\alpha)<1\]
(where ${\sf Pr}(X|Y)$ denotes conditional probability). Stronger notions of security would demand that Cath does not gain probabilistic information, so that
\[{\sf Pr}(x\in A|C,\vec\alpha)=\frac{a}{a+b}\]
(perfect $1$-security in \cite{swanson:2012}). In this paper we will only be concerned with weak security.


\section{The protocol}\label{seccolp}

In this section we shall describe the general colouring protocol, the central focus of this article. Before we do so, let us briefly describe some of the notation we shall use.

We will use $p$ to denote a prime or a power of a prime, and $\mathbb F_p$ the field with $p$ elements. If $d$ is any natural number, $\mathbb F^d_p$ denotes the vector space of dimension $d$ over $\mathbb F_p$. For a line $\ell=x+\lambda y$, we say $y$ is a {\em directing vector} of $\ell$. We denote by $\mathcal L^d_p$ the set of all lines in $\mathbb F^d_p$. The theory of finite fields and finite geometries has been extensively studied; see, for example, \cite{lidl1997,dembowski1997}.

For natural numbers $d,n$ we define $\sigma_d(n)$ to be the sum $n^{d-1}+n^{d-2}+\hdots+n^0$. We know from basic algebra that $\sigma_d(n)=(n^{d}-1)/(n-1).$ This will be a very useful quantity to keep in mind; for example,

\begin{lemm} Given $x\in \mathbb F^d_p$, there are $\sigma_{d}(p)$ distinct lines passing through $x$. 
\end{lemm}

Let $p$ be a prime or a prime power and $d \geq 2$. Assume that Alice holds as many cards as points in a line of $\mathbb F^d_p$ (i.e.\ $a = p$) and that Alice, Bob and Cath together hold as many cards as points in the whole space (i.e.\ $a+b+c = p^d$). Let $K$ be a set of $k$ colours, identified with the numbers $1,\hdots,k$. Colourings will be a crucial element in our protocol:

\begin{definition}[Colouring] \label{def.colouring} A {\em $k$--colouring} is a function $\xi:\mathcal L^d_p\to K.$

\end{definition}

\noindent That is, $\xi$ assigns a colour from $\{ 1, \, \dots \,k\}$ to each line in $\mathbb F^d_p$.

Let us now give an informal account of the protocol; Definition \ref{def.colprotocol} will give a formal version and Section \ref{line} will study its executability. It consists of the following four steps:

\begin{enumerate}

\item Alice maps all the cards into $\mathbb F^d_p$ in such a way that her cards form a line and announces the mapping.

\item Bob announces a \emph{suitably chosen} $k$-colouring $\xi$.

\item Alice announces the colour of her hand according to $\xi$.

\item Bob announces Cath's cards.

\end{enumerate}

Of course, we have yet to specify what a suitable colouring is. This will be the focus of the rest of this section. In order to guarantee that the protocol is informative (Alice and Bob can deduce the card deal after the protocol's execution), the colouring must be {\em distinguished}, as defined below. For it to be safe (Cath cannot learn any of Alice or Bob's cards), the colouring should also be {\em rich}. However, colourings which are merely rich and distinguished may give Cath too much information, and thus we will have to replace the condition of being distinguished by the stronger version of being {\em very distinguished.} Once we have defined these notions the next question is for which $a,b,c$ and $k$ such colourings exist, i.e., when the protocol is executable; this shall be the topic of later sections.

First let us focus on making the protocol be informative. Given a card deal $(A,B,C)$, we should design $\xi$ in such a way that only $A$ itself can be the line entirely contained in $A \cup C$ of whichever colour Alice announces. In this way, Bob can unequivocally identify $A$. But there may be many such spaces contained in $A\cup C$, and at the beginning of the protocol Bob has no way of telling them apart. Thus we arrive at the following:

\begin{definition}[Distinguished colouring]
We say that a $k$--colouring $\xi$ is
\emph{distinguished} for a set $E \subseteq \mathbb F^d_p$ if no two
distinct lines contained entirely in $E$ have the same colour.
\end{definition}

\begin{exam}
  \begin{figure}[htbp!]
  \centering
\begin{minipage}[c]{4.5cm}
    \includegraphics[width=4.5cm]{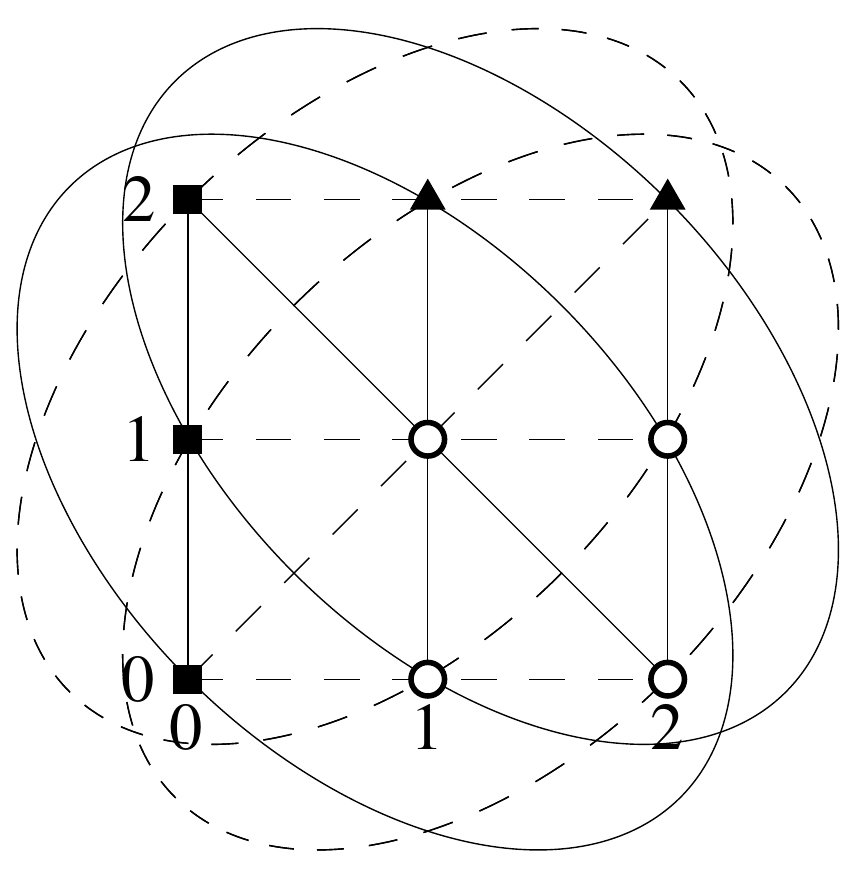}
  \end{minipage}\ \ \ 
  \begin{minipage}[c]{3cm}
    \includegraphics[width=3cm]{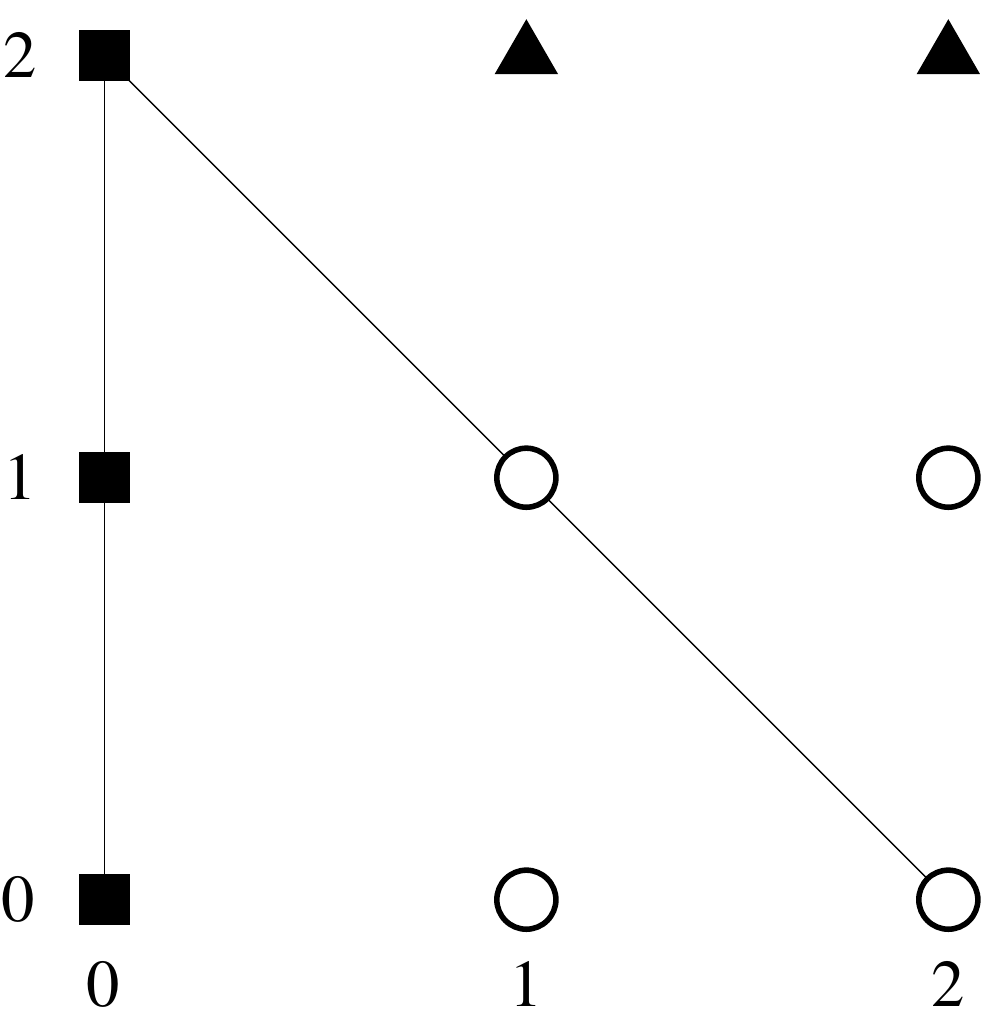}
  \end{minipage}
  \caption{A simple $2$-colouring.}

  \label{fig:NotRich}

\end{figure}
Figure~\ref{fig:NotRich} illustrates the possible effects of a distinguished colouring. The left picture shows a $2$-colouring that Bob announces after Alice maps the cards into $\mathbb F^2_3$. Throughout the examples we will continue to use black squares for Alice's cards, white circles for Bob's and black triangles for Cath's, so that cards held by Alice or Cath are black-filled shapes. Alice's cards $A$ form the line $\{00, 01, 02\}$. Cath's cards $C$ are $\{12,22\}$. The set $A\cup C$ also contains the line $\{02,12,22\}$.

In order to enable Alice, by her later response, to distinguish between the lines, Bob announces the $2$-colouring in the the left picture. It has two `colours': solid lines and dashed lines. Ellipses in the figure are lines as well, e.g. $\{02,10,21\}$ is a dashed line.

Suppose Alice announces that her line is solid. Then, Bob will learn Alice's cards and he can announce Cath's cards. But the colouring that Bob announces is not safe. After Alice announces that her cards are a solid line, Cath discards all solid lines that meet some of her points. The right picture in Figure~\ref{fig:NotRich} shows the only two lines that Cath cannot discard. This makes Cath learn that Alice has $02$ and that Bob has $10$ and $21$.
\end{exam}

Hence, colourings that are merely distinguished will make the colouring protocol informative but not necessarily safe. To remedy this, we must have enough colours, and the colouring must be rich enough, so that for every card that Cath does not hold she should consider it possible both that Alice holds it and that Alice does not hold it.

\begin{definition}[Rich colouring]\label{defrich}

A $k$--colouring $\xi$ is {\em rich} (or $c$--rich) if for any $c$-set $C$, colour $i$, and point $x \not\in C$, there is an $i$-coloured line $A$ containing $x$ that avoids $C$ and there also is an $i$-coloured line $A'$ not containing $x$ that avoids $C$.

\end{definition}

\weg{
\begin{exam}
  \begin{figure}[htbp!]
  \centering
  \includegraphics[width=3.5cm]{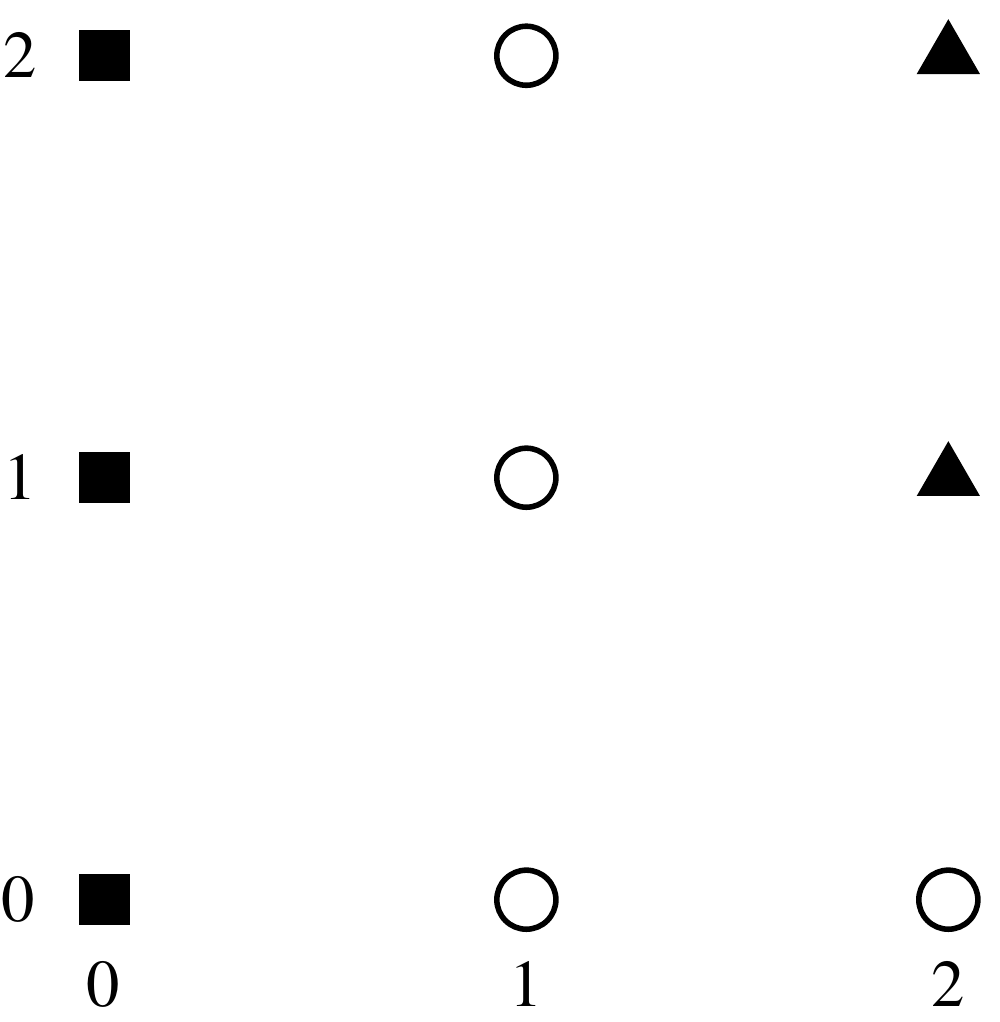}\hspace{1.5cm}
  \includegraphics[width=3.5cm]{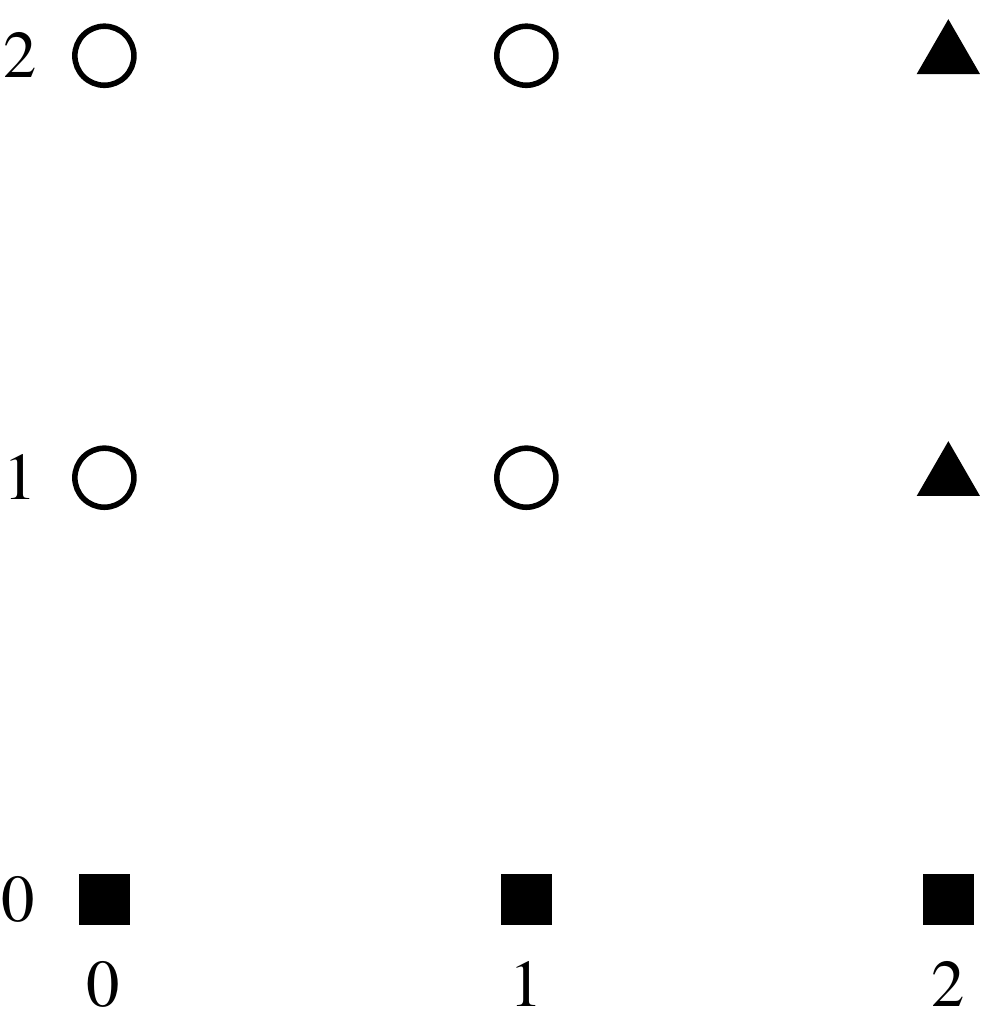}
  \caption{A rich $1$-colouring for $\mathbb F^2_3$ that is not distinguished}
  \label{fig:remark1}

\end{figure}

Figure~\ref{fig:remark1} shows another example for $\mathbb F^2_3$. The actual mapping is the left picture. Although the only line in $A\cup C$ is $\{01, 02, 03\}$, the colouring is not distinguished because another combination of five points, as on the right in the figure, contains two lines. The colouring is $2$-rich, as for every set $C$of two points ($c$--set) and every third point $x\notin C$ there are at least two lines passing through $x$ but not meeting $C$ (as through every point go four lines, two of which can be eliminated) and also at least two lines avoiding $C$ and $x$.

\end{exam}
}

\weg{
\begin{exam}\label{everything}
Consider the trivial $1$-colouring $\xi$ on the left-side image in Figure \ref{twentyfive} (all lines have the same colour). This colouring is distinguished for $A\cup C$, as this set only contains one line.

It is also rich. To see this, note that given any $c$-set $E$ and $x\not\in E$, there are five lines passing through $x$, and thus one of them avoids $E$. 

This covers one condition for richness; for the other, picking $y\not= x$ which is also not in $E$, we see that one line through $y$ avoids $\{x\}\cup E$, and thus there is a line avoiding both $x$ and $E$.
\end{exam}}

Rich and distinguished colourings are {\em almost} suitable, but we shall need an extra condition. Notice that Cath knows that Bob is to design the colouring so that it turns out to be informative. Thus, the colouring not only should be
distinguished for the actual set of cards $A \cup C$ but also for
every set of cards that Bob wants Cath to consider as possible. These will be the sets with the same {\em hue} as $A\cup C$:

\begin{definition}[Hue]
Let $\xi$ be a $k$-colouring.

Given $E,F\subseteq \mathbb F^d_p$, we write $E\approx^\xi_1 F$ if there are lines $\ell,h$ of
the same colour such that $\ell\subseteq E$, $h\cap (E\setminus
\ell)=\varnothing$ and $F=(E\setminus \ell)\cup h.$ We will say $E,F$ are {\em one swap apart}.

We then let $\approx^\xi$ be the reflexive and
transitive closure of $\approx^\xi_1$, and define a {\em hue} to be an equivalence class under $\approx^\xi$.
\end{definition}

\weg{Returning to Example \ref{everything}, our trivial one-colouring is rich and distinguished, but it has one further property; for suppose that $E$ has the same hue as $A\cup B$. Then, $E$ can only contain one line (since it has only eight points) and thus $\xi$ is {\em also} distinguished for $E$. Colourings with this property will be very useful and we shall say they are {\em very distinguished}.
}

\begin{definition}[Very distinguished colouring]
We say that a $k$--colouring $\xi$ is {\em very
distinguished} for $E \subseteq \mathbb F^d_p$ if $\xi$ is
distinguished for every $F$ of the same hue as $E$.
\end{definition}

\begin{exam}\label{exam:2colouring}
\begin{figure}[htbp!]
  \centering
  \includegraphics[width=3.5cm]{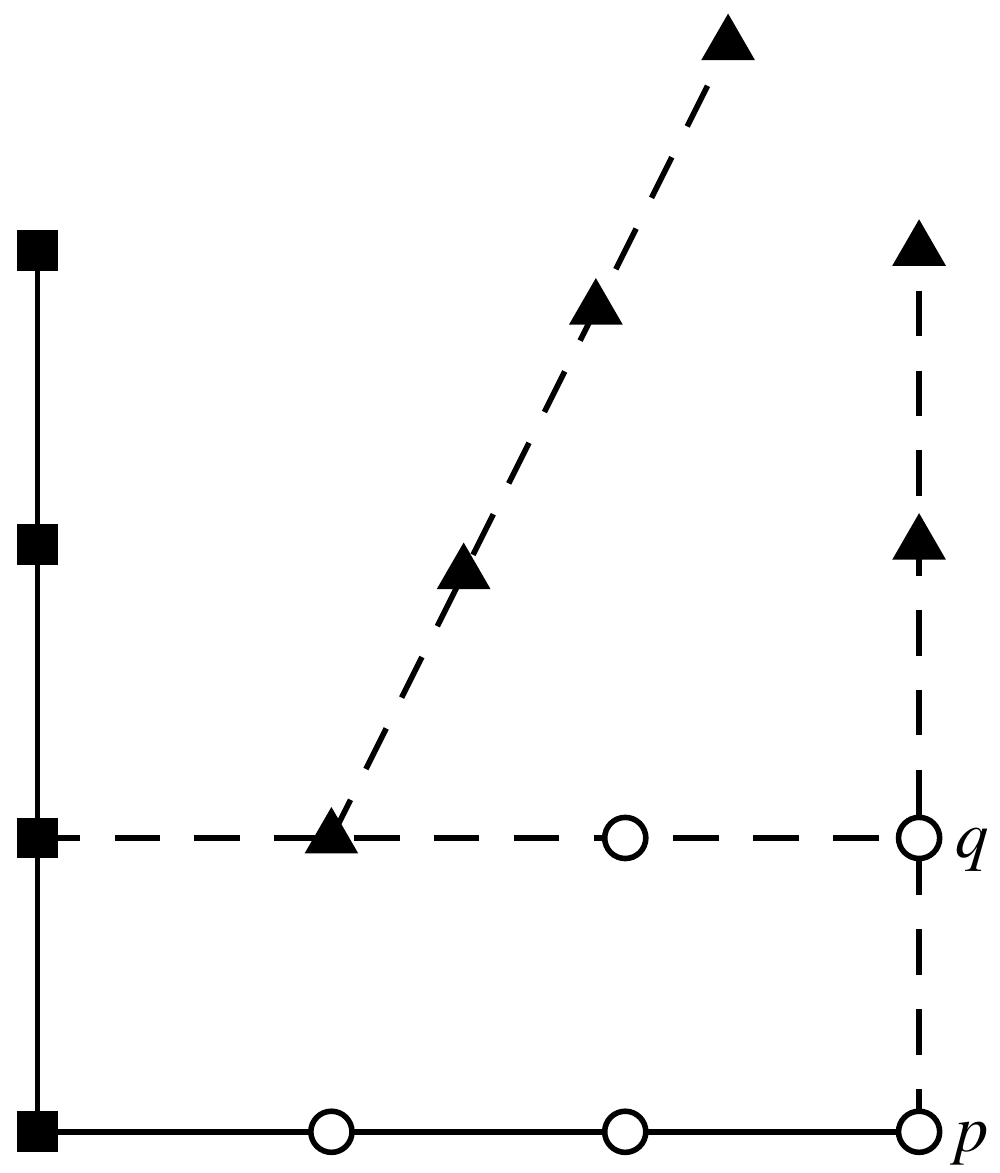}\hspace{0.5cm}
  \includegraphics[width=3.5cm]{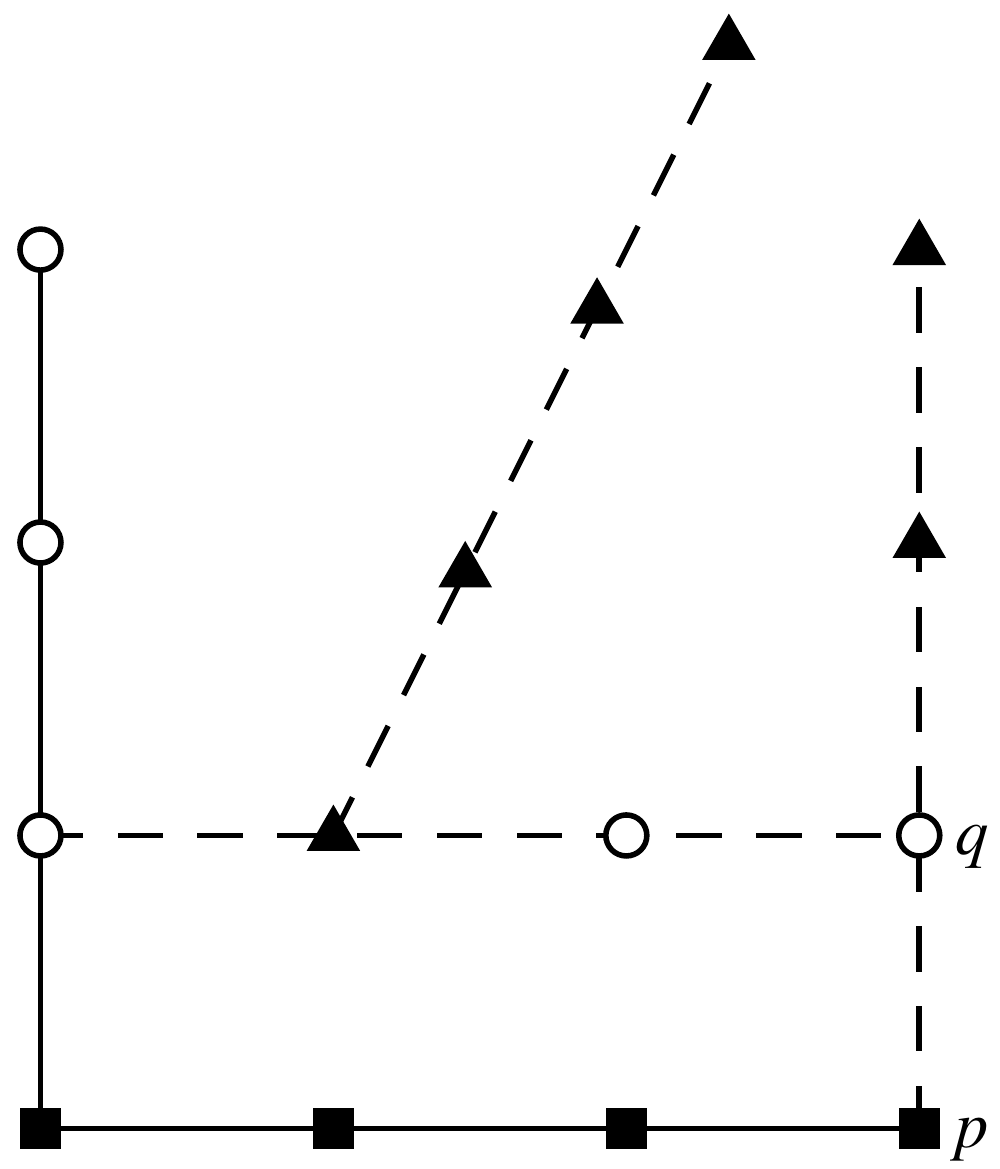}\hspace{0.5cm}
  \includegraphics[width=3.5cm]{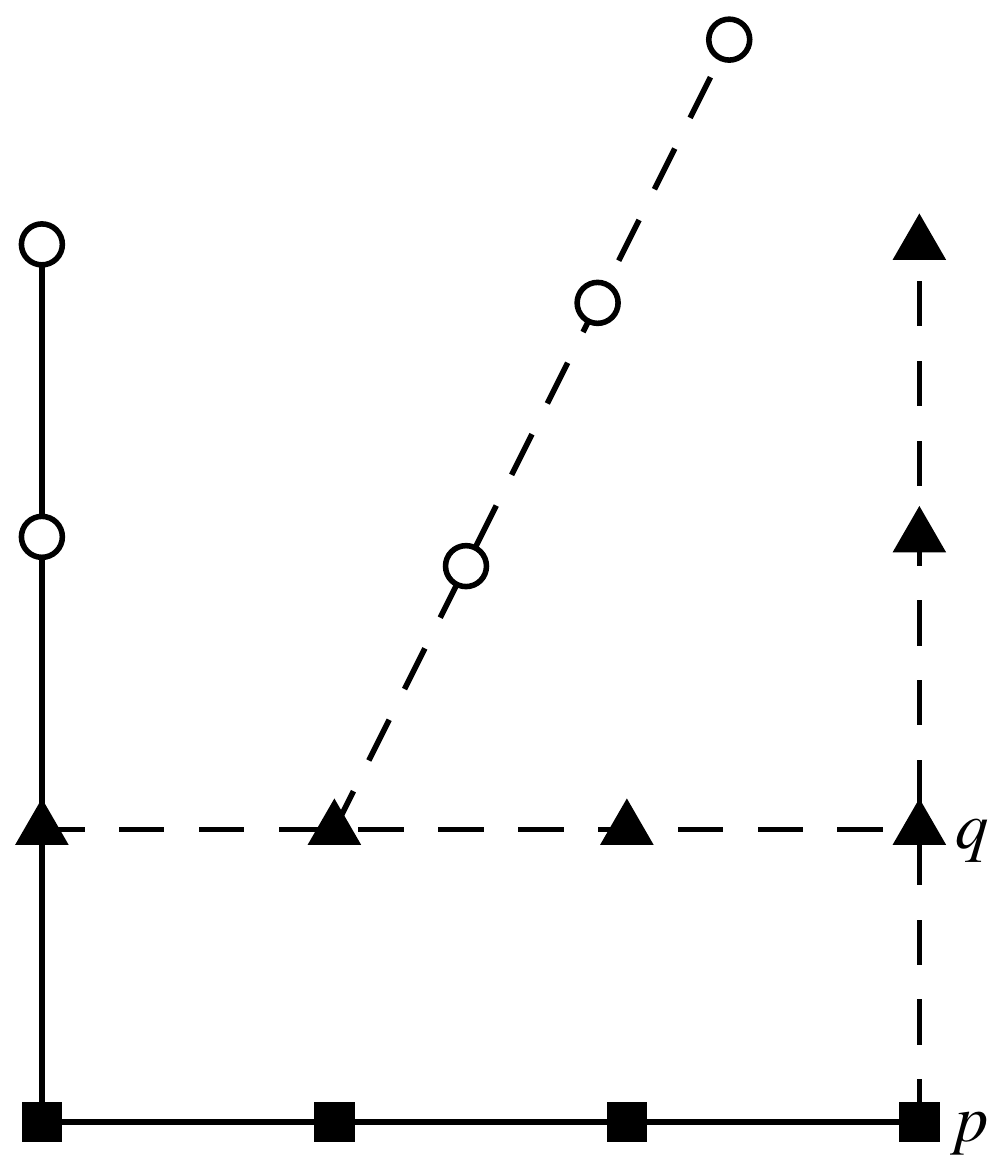}
  \caption{A $2$-colouring which is distinguished but not very
    distinguished}
  \label{fig:2colouring}
\end{figure}

Figure~\ref{fig:2colouring} shows an example of a $2$-colouring in
$\mathbb F^n_4$ for some $n$. We have represented only some points and
lines in the space. Let's suppose that the colouring is
$6$--rich. In the situation represented on the left, the
colouring is distinguished. $A\cup C$ contains two lines with
different colours. But the sets of points shown in the middle and right pictures are of the same hue as $A\cup C$, and the set of points in the right picture contains 
two lines of the same colour. Then our colouring is not distinguished in
the right picture, so it is not very distinguished in the left one.

This means that the presented $2$-colouring is not safe, even in the
left picture. Cath may learn that the points $p$ and $q$ belong to Bob,
reasoning as follows. Alice cannot have both points, because Alice's
points are aligned, but she could have one of them, as can be seen in the center picture. However, if it were the case that Alice has one point (say, $p$), Bob should
assign the line containing $p,q$ a different colour to the other two
lines in $A\cup C$, to keep Cath from learning that Bob has $q$. A similar reasoning applies to $q$, and thus Cath learns that Bob has both  $p$ and $q$.
\end{exam}

With these considerations we know which kind of colouring Bob should use: a suitably chosen colouring is rich and very distinguished, and as we shall see, the protocol so defined is safe and informative. With this we may give a formal definition of our protocol. Note that it depends on a parameter $k$ which will be specified later (see Sections \ref{line}, \ref{suitable}). Note also that the protocol is non-deterministic; recall that we are assuming all distributions to be uniform. 

\begin{definition}[Colouring protocol]\label{def.colprotocol}
Let $(a,b,c,k)$ be natural numbers such that $a$ is a prime power and for some natural number $d$, $a+b+c=a^d$. Let $D=\{1,2,\hdots,a+b+c\}$ and identify a natural number with (i) each bijection $f:D\to \mathbb F^d_a$ (ii) each $k$--colouring $\xi$, (iii) each of the $k$ colours and (iv) each element of $D\choose c$.

Let $(A,B,C)$ be a deal of size $(a,b,c)$. Then, the protocol is as follows.
\begin{enumerate}

\item Alice randomly chooses a bijection $f:D\to\mathbb F^d_a$ in such a way that her cards form a line. She announces $f$.

\item Bob randomly chooses a $k$--colouring $\xi$ which is rich and very distinguished  for $A\cup C$ and announces $\xi$.

\item Alice announces the colour of her hand.

\item Bob announces Cath's cards $C$.

\end{enumerate}

\end{definition}

Strictly speaking, when we state that ``Alice announces $f$'' it is understood that she announces the natural number assigned\footnote{Alice's first announcement does not correspond one-one to sets of hands that she may hold, for many mappings may define the same set of lines (just compose $f$ with an invertible linear transformation). A similar observation applies to Bob's colourings. This is why we do not model announcements as sets of hands, as is done, for example, in \cite{albertetal:2005}.} to the function $f$, and similarly to other steps in the protocol. Thus an execution can be represented in the form
\[(Alice,f),(Bob,\xi),(Alice,n),(Bob,X)\]
where $f:D\to \mathbb F^d_a$, $\xi$ is a colouring, $n$ a colour and $X\in {D\choose c}$. Since the players take turns we may omit the agents and write $(f,\xi,n,X)$.

With this in mind, let us see that the colouring protocol is indeed a protocol as defined in Section \ref{sec.some}. Intuitively, we must show that every announcement depends only on information available to the respective player.

\begin{lemm}
The colouring protocol satisfies Definition \ref{defprot}.
\end{lemm}

\begin{proof}
Denote the colouring protocol by $(j,\pi)$, where $j(\vec\alpha)=Alice$ if $\vec\alpha$ has even length and otherwise $j(\vec\alpha)=Bob$ and $\pi(\delta,\vec\alpha)$ is the set of possible announcements as specified by Definition \ref{def.colprotocol}. We need to verify that $\pi$ is invariant under $\stackrel{j(\vec\alpha)}\sim$. Let us do this step by step:

\paragraph{1} {\em Alice announces a map $f:D\to\mathbb F^d_p$ in such a way that her cards form a line.} Suppose $(A,B,C)\stackrel{Alice}\sim(A',B',C')$ and $f\in\pi((A,B,C),())$, so that $A=A'$ and $f$ maps $A$ into a line $\ell$. Then, $f(A')=f(A)=\ell$ and hence $f\in\pi((A',B',C'),())$, as required.

\paragraph{2} {\em Bob announces a rich and very distinguished $k$--colouring $\xi$.} Suppose that $(A,B,C)\stackrel{Bob}\sim(A',B',C')$ and $\xi\in\pi((A,B,C),(f))$, so that $\xi$ is a rich and very distinguished colouring for $f(A\cup C)$. But since $B=B'$ we also have $A\cup C=A'\cup C'$ and thus $\xi$ is also rich and very distinguished for $f(A'\cup C')$, so that $\xi\in \pi((A,B,C),(f))$, as required.

\paragraph{3} {\em Alice announces the colour of her hand according to $\xi$.} As before, suppose that $(A,B,C)\stackrel{Alice}\sim(A',B',C')$ and $n\in\pi((A,B,C),(f,\xi))$. Then, $\xi(f(A'))=\xi(f(A))=n$ so we have that $n\in\pi((A,B',C'),(f,\xi))$.

\paragraph{4} {\em Bob announces Cath's cards.} Once again suppose that $(A,B,C)\stackrel{Bob}\sim(A',B,C')$. By our specification $X\in\pi((A,B,C),(f,\xi,n))$ if and only if $X=C$, so we must show that if $((A',B,C'),(f,\xi,n))$ is an execution of the protocol then $C=C'$.

Since $\xi$ is distinguished there is a unique line $\ell\subseteq f(A\cup C)=f(A'\cup C')$ with $\xi(\ell)=n$. Since Alice's hand forms a line, we have $f^{-1}(\ell)=A=A'$. Thus $C=(D\setminus B)\setminus A=(D\setminus B)\setminus A' =C'$, as required.
\end{proof}

It remains to check that the colouring protocol indeed provides a solution to the generalized Russian cards problem:

\begin{theo}\label{ifex} The colouring protocol is safe and informative, provided it is
executable.
\end{theo}

\proof  The protocol is obviously informative given Bob's last announcement, so we focus on safety. Let $\vec\alpha=(f,\xi,n,C)$ be an execution of the colouring protocol.

First pick $x$ that Cath does not hold. Because $\xi$ is rich,
there is a line $\ell$ with colour $i$ passing through $f(x)$ and not meeting
$C$. Further, note that $f(A) \cup f(C) \approx^\xi_1 \ell \cup f(C)$, so that
$\xi$ is also very distinguished for $\ell\cup f(C)$. Then, setting $A'=f^{-1}(\ell)$ we see that $((A',B',C),\vec\alpha)$ is also an execution of the protocol, but $x\in A'$.

The argument that Cath also considers it possible that Alice does not hold $x$ is very similar.
Since the set $C$ has exactly $c$ elements, there is a line $h$ with colour $i$ not
meeting $f(C)\cup\cbra f(x)\cket$, so that setting $A''=f^{-1}(h)$ and reasoning as above, $((A'',B'',C),\vec\alpha)$ is another execution of the protocol where Bob holds $x$.

Thus the colouring protocol is safe.
\endproof

We have so far worked under the assumption that the protocol is executable. The rest of the paper will be devoted to examining when this is the case, and determining the parameter $k$.

\section{Executability}\label{line}

In what follows we will look for conditions to guarantee the existence of a rich and very distinguished $k$--colouring for some number of
colours $k$. First, let us introduce the notion of {\em density:}

\begin{definition}[Density]
Say a $k$--colouring $\xi$ has {\em density} $m$
if, given a point $x \in \mathbb F_p^d$ and a colour $i$, there are
at least $m$ $i$--coloured lines through $x$.
\end{definition}

There is a close connection between density and richness; to be precise, a colouring that is dense enough is automatically rich.

\begin{lemm}\label{notmeeting}
If $\xi$ is a $k$--colouring of density $c+2$, then $\xi$ is
rich.
\end{lemm}

\proof Let $E$ be a subset of $\mathbb F_p^d$ with $c$ elements. Fix
a colour $i$ and a point $x \not \in E$. We note that if $\ell,h$ are
two distinct lines passing through $x$, then $\ell\cap h=\cbra
x\cket$. Therefore, $E$ contains the disjoint union
\[\bigcup \{E\cap \ell:x\in\ell\text{ and }\xi(\ell)=i\}.\]
It follows that for some $\ell$ with $\xi (\ell)=i$ and $x\in \ell$,
$E\cap \ell$ must be empty, otherwise $E$ would have at least $c+2$
elements.

Similarly, if we pick $y\not\in E$ different from $x$, we see that there is an $i$-coloured line through $y$ not meeting $E\cup\{x\}$, satisfying the second requirement.
\endproof

Thus in order to construct rich colourings, we may focus on constructing dense colourings; this is not too difficult, as witnessed by the following:

\begin{lemm}\label{knitlemm}
Let $\ell_1,...,\ell_k$ be distinct lines of $\mathbb
F^d_p$ and assume that $\sigma_d(p) \geq k(m+1).$ Then, there is a
$k$--colouring $\xi$ of density $m$ such that for $i\leq k$,
$\xi(\ell_i)=i$.
\end{lemm}

\proof There are $\sigma_d(p)\geq k(m+1)$ non-collinear vectors in
$\mathbb F_p^d$, and hence we can pick a set $D$ of $mk$ non--collinear
vectors which are not directing vectors of any of the lines
$\ell_i$. Partition $D$ into $k$ disjoint sets $D_i$ of $m$
elements. Then, given a line $h \in \mathbb F_p^d$, put $\xi(h)=i$ if
either $h=\ell_i$ or $h$ has directing vector in $D_i$ for some $i
\leq k$. Otherwise put, for instance, $\xi(h)=1$. It is easy to see
that $\xi$ satisfies the desired properties.
\endproof

This will be sufficient for finding rich colourings. Now our goal is to construct a very distinguished $k$--colouring for $A \cup C$. As with richness, we will do so by introducing a stronger, approximate notion -- that of a {\em perfect colouring}. It is not easy to tell under which conditions very distinguished colourings exist or how one may go about finding one, but finding perfect colourings will be straightforward.

The notion of a perfect colouring has the disadvantage that perfection is not hue-invariant. To deal with this issue, we will first need an intermediate concept: that of {\em critical colourings.} Every perfect colouring is critical and every critical colouring is very distinguished. Perfect colourings are the easiest to identify, but critical colourings are easier to work with than either perfect or very distinguished colourings.

Let us then begin by defining critical colourings:

\begin{definition}[Critical colouring]\label{def:criticalCol}
Given $E\subseteq \mathbb F_p^d$, we say that a $k$-colouring $\xi$
is \emph{critical} for $E$ if there exists a set $L =\{ \ell^\ast_1,
\, \dots, \, \ell^\ast_n \}$ of different colours such that,
for every line $h\subseteq \mathbb F_p^d$, we have
\[\left|(h\cap E)\setminus \displaystyle\bigcup_{i\leq n}\ell^\ast_i\right|<p-k,\]
 where $|S|$ stands
for the cardinality of a set $S$. We will say the lines in $L$ are
{\em $\xi$-critical lines} for $E$.
\end{definition}

The notion of critical colouring in Definition~\ref{def:criticalCol}
captures an important difference between the $2$-colourings in 
Figures~\ref{fig:2colouring} and~\ref{fig:2colVery}. The $2$-colouring
in Figure~\ref{fig:2colouring} is not critical for $E=A\cup C$, as for any set of lines $L$ that we can select, there is a line $h$ such that \[\left|(h\cap E)\setminus \bigcup
L\right| \geq 2=p-k,\]
as the reader may verify by examination.

On the other hand, the colouring in the left-hand side of Figure~\ref{fig:2colVery} is
critical for $E=A\cup C$, the union of the two lines and the fragment. We select $L$ as
the two complete lines in $E$. Then, for any line $h$,
the property of Definition~\ref{def:criticalCol} holds. In particular, the line $h$
with the three holes is the line outside of $L$ with most points in $E$. For that line,
\[\left|(h\cap E)\setminus \bigcup L\right| =p-3 < p-2.\]

  Recall Example~\ref{exam:2colouring}, where we showed that the colouring in Figure~\ref{fig:2colouring} is   unsafe. Intuitively, the insecurity of a $k$-colouring (in  Figure~\ref{fig:2colouring} we have $k=2$) is produced because $A\cup
  C$ contains $k$ lines, along with a line fragment $\ell$ with $k$ or less gaps -- observe the line missing $p$ and $q$ in
  Figure~\ref{fig:2colouring}. Then, though the $k$-colouring is
  distinguished, it might not be very distinguished, as we may be able to use several swaps to move the $k$ lines and fill all the gaps in $\ell$, thus generating $k+1$ lines and repeating a colour.

However, when $A\cup C$ contains $k$ or less lines and all fragments have
more than $k$ points missing, then a distinguished $k$-colouring $\xi$ is always
very distinguished. Figure~\ref{fig:2colVery} shows an example of this
situation.

\begin{figure}[htbp!]
  \centering
  \includegraphics[width=4.5cm]{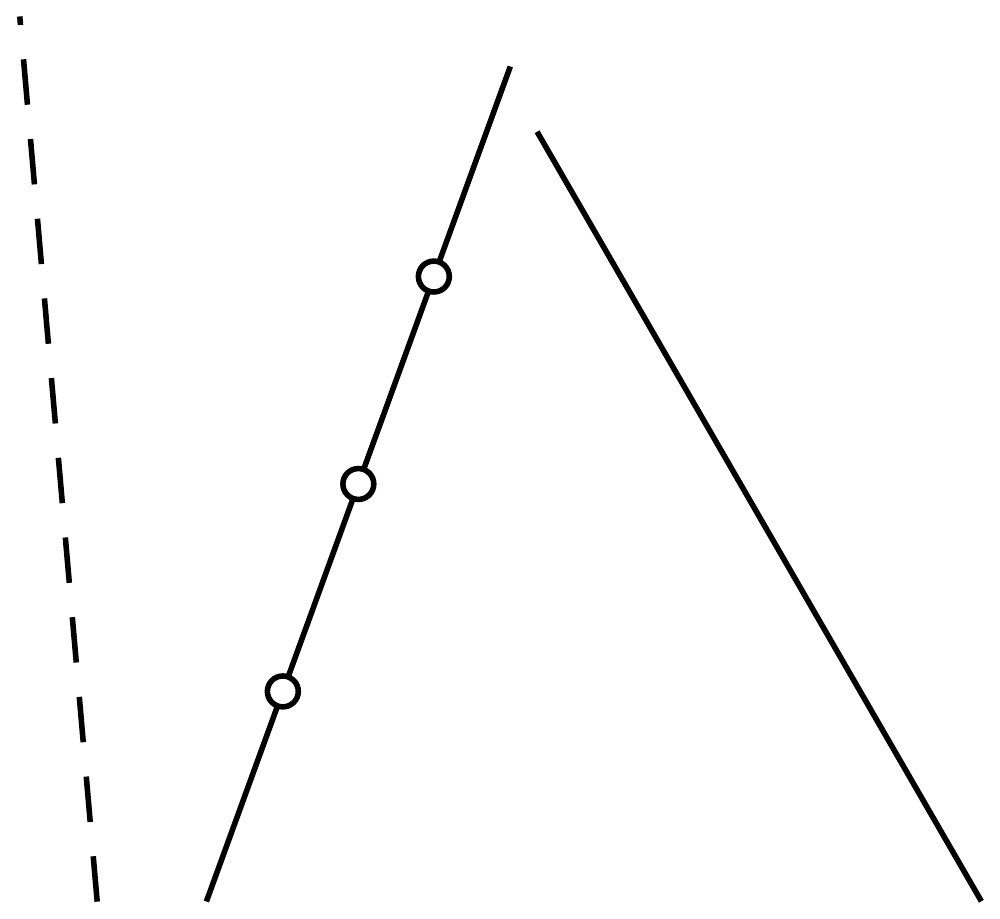}\hspace{0.2cm}
  \includegraphics[width=4.5cm]{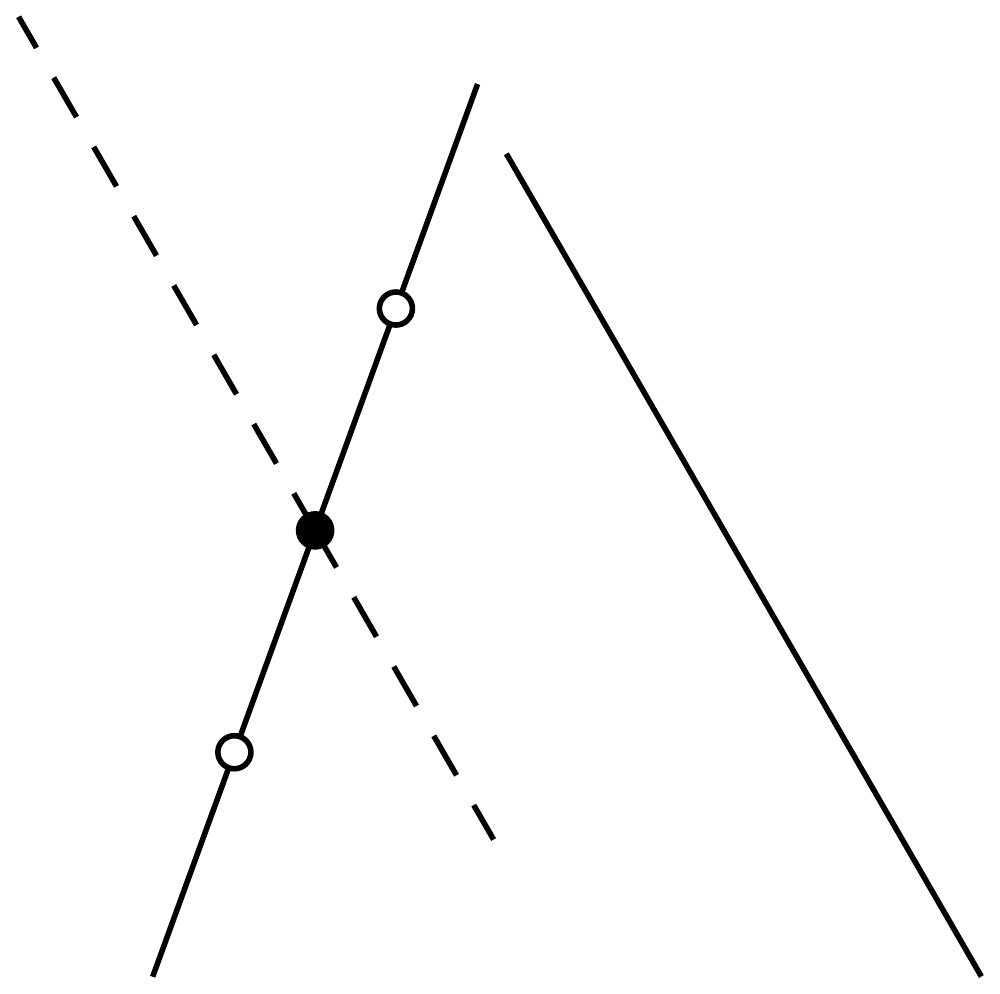}\hspace{0.2cm}
\includegraphics[width=2.5cm]{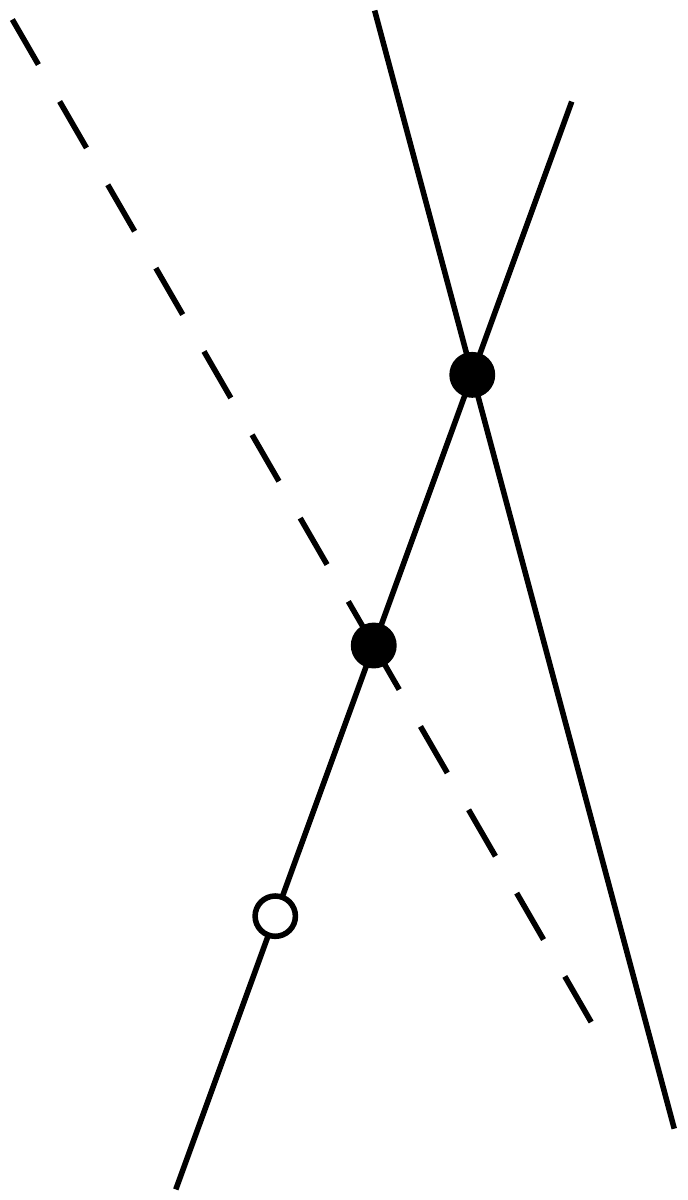}
  \caption{A critical $2$-colouring}
  \label{fig:2colVery}
\end{figure}

\begin{exam} In the left-hand side of Figure~\ref{fig:2colVery}, $A\cup C$ contains two lines and a fragment with three or more points missing. Alice has one of these
lines, and Cath has the other line plus the segment. Then, a swap corresponds to moving a line to an empty position of the same colour, and it is not
possible to arrange the two lines in $A\cup C$ in such a way that they fill
the three gaps in the fragment. Thus there is no configuration of the same hue for which $\xi$ is not distinguished, so that $\xi$ is very distinguished.
\end{exam}

The above considerations suggest that critical colourings are very distinguished, and indeed this will turn out to be the case. Before we show this, let us make a simple observation:

\begin{lemm}\label{isexc}
If $\xi$ is a critical $k$--colouring for $E$ and $\ell\subseteq E$,
then $\ell$ is a $\xi$--critical line for $E$, independently of how the other critical lines are chosen.
\end{lemm}

\proof Let $L=\{\ell^\ast_1 , \dots , \ell^\ast_n\}$ be $\xi$-critical
lines for $E$. Notice that since they are lines of different colours,
we have $n \leq k$. If $\ell$ is not $\xi$--critical then for every
$i\leq n$, $\ell\cap \ell^\ast_i$ is either empty or a singleton. It
therefore follows that

\[\left|(\ell\cap E)\setminus\bigcup L\right| =
\left| \ell\setminus \displaystyle\bigcup_{i\leq n}\ell^\ast_i
\right| \geq p - n \geq p-k,\]
which contradicts the definition of a critical $k$--colouring.
\endproof

Now we will check that critical colourings are very distinguished. For this, it suffices to show that they are distinguished and that a critical colouring for $E$ is critical for the entire hue of $E$; the former is straightforward, the latter rather involved.

\begin{lemm}\label{placeisnet}
Every critical $k$--colouring for $E$ is distinguished for $E$.
\end{lemm}

\proof From Lemma \ref{isexc} it follows that if $\ell\subseteq E$
then it is critical; but there is at most one critical line of each
colour so there cannot be $h\not=\ell$ of the same colour both
contained in $E$.
\endproof

\begin{lemm}\label{break}
Suppose $\xi$ is a critical $k$-colouring for $E$ with $k< p$ and $F$ has the same hue as $E$. Then,
$\xi$ is also a critical colouring for $F$.
\end{lemm}

\proof We shall prove this by induction on the number of
swaps between $E$ and $F$. The base case ($F=E$) is vacuous. For the inductive
step, suppose $F=(G\setminus \lambda_1)\cup \lambda_2$ with
$\xi(\lambda_1)=\xi(\lambda_2)=j$ and $G$ one swap closer to $E$ than $F$ is, so that by induction we may assume that there are $\xi$-critical
lines $\ell^\ast_1,\hdots,\ell^\ast_n$ for $G$ with $n \leq k$. Observe that $\ell^\ast_j=\lambda_1$; indeed, $\lambda_1\subseteq G$, so that by Lemma \ref{isexc}, $\lambda_1=\ell^\ast_i$ for some $i$. But $\xi(\lambda_1)=j$, and therefore $i=j$.

Our goal is to show that $\xi$ is also critical for $F$; we will do this by proving that the lines $h^\ast_1,\hdots,h^\ast_n$ defined by $h^\ast_i=\ell^\ast_i$ for $i\not=j$ and
$h^\ast_j=\lambda_2$ are critical. For this it will suffice to check that if $\ell$ is any line,
\begin{equation}\label{claim}
(\ell\cap F)\setminus\bigcup_{i\leq n}h^\ast_i\subseteq (\ell\cap G)\setminus\bigcup_{i\leq n}\ell^\ast_i;
\end{equation}
this inclusion would then imply that
\[\left|(\ell\cap F)\setminus\bigcup_{i\leq n}h^\ast_i\right|\leq\left|(\ell\cap G)\setminus\bigcup_{i\leq n}\ell^\ast_i\right|<p-k,\]
showing that $h^\ast_1,\hdots,h^\ast_n$ (and hence $\xi$) are critical for $F$.

To establish (\ref{claim}), pick $x\in(\ell\cap F)\setminus \bigcup_{i\leq n}h^\ast_i$. It is obvious that $x\in \ell\cap G$, except perhaps in the case that $x\in \lambda_2$. But $\lambda_2=h^\ast_j$, so this cannot be, and therefore we always have $x\in \ell\cap G$. 

Next, we need to check that $x\not\in \ell^\ast_i$ for any $i\leq k$. This is obvious if $i\not=j$, since $\ell^\ast_i=h^\ast_i$. Thus it remains to rule out that $x\in\ell^\ast_j$. If we had $x\in\ell^\ast_j=\lambda_1$, we would also have that $x\in \lambda_1\setminus h^\ast_j=\lambda_1\setminus \lambda_2$. But $F$ does not intersect this set, so this is impossible.

We conclude that (\ref{claim}) holds, and thus $\xi$ is critical for $F$, as desired.
\endproof

As promised, we now have the following:

\begin{lemm}\label{critisdis}
Every critical colouring is very distinguished.
\end{lemm}

\proof
If $\xi$ is critical for $E$, then by Lemma \ref{break}, it is also critical for every set of the same hue as $E$; thus, by Lemma \ref{placeisnet}, it is also distinguished for every such set.
\endproof

We have seen that in order to construct very distinguished colourings, it suffices to construct critical colourings. This has the advantage that the notion of being critical depends only on a set $E$ and not on the entire hue of $E$. However, it is still not entirely obvious when exactly critical colourings exist or how one is to go about building one.

This is where perfect colourings come into play. Although perfect colourings are not perfect for their entire hue, they are very easy to identify, and every perfect colouring is automatically critical.

\begin{definition}[Perfect $k$--colouring]\label{def:perfectCol}
Given $E\subseteq \mathbb F_p^d$, let $L_m(E)$ denote the set of
all lines $\ell$ such that $|\ell\cap E|\geq m$.

We say that a
$k$-colouring $\xi$ is {\em perfect} for $E$ if different elements of $L_{p-k}(E)$ have different colours.
\end{definition}

Perfect colourings, when they exist, are quite easy to identify and construct, yet they are always very distinguished:

\begin{prop}\label{isasafe}
If $\xi$ is a perfect colouring for $E$, then $\xi$ is very distinguished for $E$.
\end{prop}

\proof
First we note that $\xi$ is critical, since we can take all of $L_{p-k}(E)$ as critical lines, given that they are all of different colour. Then, by Lemma \ref{critisdis}, we have that $\xi$ is very distinguished, as required.
\endproof

Now for the main result of this section:

\begin{theo}\label{mainline}
Assume that $a,b,c,d,k$ satisfy the following conditions:
\begin{enumerate}
\item $a$ is a prime or a prime power,
\item $a+b+c=a^d$,
\item $k<a$,
\item for every $S\subseteq\mathbb F^d_a$ with $|S| \leq a+c$, we have $|L_{a-k}(S)|\leq k$;
\item $\sigma_d(a)\geq k(c+3)$.
\end{enumerate}
Then, the colouring protocol is executable.
\end{theo}

\proof
Assume that Alice has announced $f$ and for this proof let $A,B,C$ denote the images of each player's hand under $f$. By assumption, $|L_{a-k}(A\cup C)|\leq k$, so Bob can enumerate $L_{a-k}(A\cup C)$ by
 $\vec \ell=\langle\ell_1,\hdots,\ell_n\rangle$, with $n\leq k$. Then, by Lemma \ref{knitlemm},
 there exists a colouring $\xi$ with density $c+2$ and such that $\xi(\ell_i)=i$, which by construction is perfect for $A\cup C$, so that by Proposition \ref{isasafe}, it is also very distinguished. Further, by Lemma \ref{notmeeting}, we see that $\xi$ is rich, as needed.
\endproof

The above proof should be seen as a {theoretical argument} that very distinguished colourings exist under the above conditions -- but never as an {algorithm} for constructing them. It is essential for our protocol that Bob choose randomly among all possible colourings. We leave the specification of feasible algorithms for doing so for future work.

\section{Finding suitable parameters}\label{suitable}

The conditions given by Theorem \ref{mainline}, while rather general, remain somewhat implicit. In this subsection we shall compute some explicit bounds on the parameters which guarantee that these conditions are met. The computations we will make are a bit rough, but will nevertheless give us a large family of parameters for which the protocol is guaranteed to work.

They are based on the following counting lemma:

\begin{lemm}\label{lastlemm}
If a set $E\subseteq\mathbb F^d_p$ is such that
\[|E|< (k+1)(p-k)-\frac {k(k+1)}2,\]
then $|L_{p-k}(E)|\leq k.$
\end{lemm}

\proof
We argue by contrapositive, assuming that $| L_{p-k}(E)|> k.$

Let $\ell_1,\hdots,\ell_{k+1}$ be distinct lines such that
$|E\cap \ell_i|\geq p-k$. Then, we have
\[\begin{array}{lcl}
|E|&\geq&\left|\bigcup_{i\leq k+1}(\ell_i\cap E)\right|\\
&\geq& \sum_{i\leq k+1}|\ell_i\cap E|-\sum_{i<j\leq k+1}|\ell_i\cap \ell_j|\\
&\geq& (k+1)(p-k)-\frac12{k(k+1)}.
\end{array}
\]
\endproof

Thus in view of Theorem \ref{mainline}, in order to find suitable parameters, it suffices to solve the system of inequalities
\begin{enumerate}
\item $a+c<(k+1)(a-k)-\frac 12 k(k+1)$,
\item $\displaystyle\dfrac{a^d-1}{a-1}\geq k(c+3)$
\end{enumerate}
or, simplifying a bit,
\begin{enumerate}
\item $c<ak-\frac 32 k(k+1)$,
\item ${a^d}> k(a-1)(c+3).$
\end{enumerate}
From this we immediately obtain the following asymptotic result:

\begin{theo}\label{mainbound}
If $a$ is a large enough prime power, the colouring protocol is executable with
\begin{enumerate}
\item $c<O(a^{\nicefrac 32})$ and $d=3$ or
\item  $c<O(a^2)$ and $d=4$.
\end{enumerate}
\end{theo}

\proof
In view of Theorem \ref{mainline} and Lemma \ref{lastlemm}, it suffices to show that the above inequalities hold. 

For the first result, set $k\approx\sqrt a$ and $c\approx\nicefrac{a^{\nicefrac 32}}2$. Then, we have that $ak-\nicefrac{3k(k+1)}2=a^{\nicefrac 32}+O(a),$ which for large $a$ is greater than $\nicefrac{a^{\nicefrac 32}}2$. Meanwhile, $k(a-1)(c+3)=\nicefrac{a^3}2+O(a^2),$ which for large $a$ is bounded by $a^3$.

The second is similar; here, set $k\approx\nicefrac a2$ and $c\approx \nicefrac{a^2}9.$
\endproof

A nice conclusion we obtain from the above result is that indeed $c$ can be larger than $a$ by any order of magnitude we desire, provided $a$ is large enough:

\begin{cor}
Given any natural number $N$, there exist $a,c$ such that $\nicefrac ca>N$ and the colouring protocol is executable, sound and informative for $(a,a^3-a-c,c)$.
\end{cor}

\proof
If $a\gg N^2$ then ${a^{\nicefrac 32}}=({a^{\nicefrac 12}})a\gg Na$ and, by Theorem \ref{mainbound}.1, the required inequalities hold and the protocol is executable.
\endproof

Note that we may use Theorem \ref{mainbound}.1 or Theorem \ref{mainbound}.2 depending on whether
 we wish to keep $b$ relatively small with respect to $a$ or $a$ relatively small with respect to $c$.
 In either case, $b=O(c^2)$.

\section{Conclusions and further research}

The colouring protocol we have presented gives a new and flexible solution to the generalized Russian cards problem. Our protocol solves the problem in many cases where the eavesdropper has more cards than one of the players, and is the first known solution to achieve this. If the generalized Russian cards problem is to be understood as {\em Find triples $(a,b,c)$ for which a safe and informative protocol of any length exists,} then the current work is a giant stride over what had been previously achieved.

There are many variations that could be analyzed and might yield important results. Alice may have more than one line, there may be more players, Cath may be allowed to learn a few of the cards, etc. Along these lines, a particularly promising direction would be to replace lines by other algebraic curves. This could, potentially, reduce the size of the whole space (and hence Bob's hand) without compromising the existence of suitable colourings.

A different direction to pursue involves cases where either $a$ or $a+b+c$ are not prime powers. There are already standard techniques to deal with these; one finds a prime which is not much larger than the desired parameter and works with that instead. Such techniques have already been used to extend a different protocol to many new triples in \cite{cordonetal:2012} but have not been worked out in our context.

Meanwhile, the generalized Russian cards problem may have a recreational flavour to it, but it has potential for serious applications in secure communication. These protocols have an advantage over most traditional encryption methods in that they are ``unconditionally secure'', meaning that an eavesdropper is unable to decipher messages even when granted unlimited computational capacity. The presupposition of a ``dealing stage'' is not so different to authentication systems using security tokens, where the user receives a physical device that generates one-time passwords. A possible application is the development of challenge-response authentication systems. The user Bob is authenticated against the server Alice by running a Russian cards protocol. If Alice checks that Bob's responses are coherent with the key he should have (Bob's hand) then the authentication succeeds.

Real-world implementation of these protocols would also require an additional computational and probabilistic analysis. What is the complexity of running the colouring protocol? Are there good algorithms for finding suitable colourings? Even if Cath does not learn any cards, how do we minimize the probability that she guesses them correctly?

To summarize, there is much to be done indeed!

\section*{Acknowledgements}

This research was supported by the project {\em Logics for Unconditionally Safe Protocols}, Excellence Research Project of the Junta de Andaluc\'ia P08-HUM-04159.




\bibliographystyle{plain}







\end{document}